\documentclass[11pt,letterpaper]{article}
\usepackage[plain]{fullpage}
\usepackage{latexsym}
\usepackage{amsmath}
\usepackage{amsthm}
\usepackage{amssymb}
\usepackage{xcolor}
\usepackage{url}
\newtheorem{theorem}{Theorem}[section]
\newtheorem{lemma}[theorem]{Lemma}
\newtheorem{proposition}[theorem]{Proposition}
\newtheorem{corollary}[theorem]{Corollary}

\newtheorem{fact}[theorem]{Fact}

\def\Z{{\mathbb Z}}

\def\F{{\mathbb F}}

\newcommand{\poly}{\mathrm{poly}}

\newcommand\Mult{{\mathrm{Mult}}}
\newcommand\Inv{{\mathrm{Inv}}}
\newcommand\Id{{\mathrm{Id}}}
\newcommand\Add{{\mathrm{Add}}}
\newcommand{\DLOG}{\mbox{\rmfamily\textsc{DLOG}}}
\newcommand{\CDH}{\mbox{\rmfamily\textsc{CDH}}}
\newcommand{\DDH}{\mbox{\rmfamily\textsc{DDH}}}
\newcommand{\fuc}{\mbox{\rmfamily\textsc{DH}}}
\newcommand{\US}{\mbox{\rmfamily\textsc{US}}}

\newcommand{\QC}{\mbox{\rmfamily\textsc{Q}}}
\newcommand{\RC}{\mbox{\rmfamily\textsc{R}}}
\newcommand{\SECRET}{\mbox{\rmfamily\textsc{SECRET}}}

\newcommand{\suppress}[1]{}

\title{Discrete logarithm and Diffie-Hellman problems in identity
  black-box groups}

\author{ 
G\'abor Ivanyos \thanks{ Institute for Computer Science and Control, Budapest, Hungary 
({\tt Gabor.Ivanyos@sztaki.hu}).} 
\and 
Antoine Joux\thanks{
CISPA Helmholtz Center for Information Security, Saarbr\"ucken, Germany and
Sorbonne Universit\'e, Institut de Math\'ematiques de
Jussieu--Paris Rive Gauche, CNRS, INRIA, Univ Paris Diderot,
Campus Pierre et Marie Curie, 
F-75005 Paris, France
({\tt antoine.joux@m4x.org}).}
\and 
Miklos Santha \thanks{ CNRS, IRIF, Universit\'e de Paris, F-75013 Paris, France;  and Centre for Quantum Technologies and MajuLab, National University of Singapore, 
Singapore 117543
({\tt miklos.santha@gmail.com}).}
}

\begin{document}

\maketitle

\begin{abstract}
{

  We investigate the computational complexity of the discrete
  logarithm, the computational Diffie-Hellman and the decisional
  Diffie-Hellman problems 
  in some identity black-box groups $G_{p,t}$, where $p$ is a prime
  number and $t$ is a positive integer.  These are defined as quotient
  groups of vector space $\Z_p^{t+1}$ by a hyperplane $H$ given
  through an identity oracle. While in general black-box groups 
that have
  unique encoding 
of their elements
these computational problems are classically all
  hard and quantumly all easy, we find that in the groups $G_{p,t}$
  the situation is more contrasted.  We prove that while there is a
  polynomial time probabilistic algorithm to solve 
  the decisional Diffie-Hellman problem in $G_{p,1}$, the
  probabilistic query complexity of all the other problems is
  $\Omega(p)$, and their quantum query complexity is
  $\Omega(\sqrt{p})$.  Our results therefore provide a new example of
  a group where the computational and the decisional Diffie-Hellman
  problems have widely different complexity.
}

\end{abstract}





\section{Introduction}

Black-box groups 
were
introduced by Babai and Szemer\'edi \cite{BabSzem}
for studying the structure of finite matrix groups. In a black-box group,
the group elements are encoded by binary strings of certain length,
the 
group operations and their inverses
are given by oracles.
Similarly, identity testing, that is checking whether an element is equal to
the identity element, is also done with a special identity oracle. 
These oracles are also called black-boxes, giving their names to the groups. 
Identity testing
is required when several strings may encode the same group
element. In this case we speak about non-unique encoding, in opposition
to unique encoding when every group element is encoded by a unique string.
Black-box groups with non-unique encoding are motivated
by their ability to capture factor groups of subgroups of matrix groups by 
certain normal subgroups which admit efficient algorithms for
membership testing. An important example for such a subgroup 
is the solvable radical,
that is the largest solvable normal subgroup. 

A black-box group problem that concerns some global
properties of the group may have no inputs other than
the oracles.
In some other cases 
it might also have standard inputs 
(or inputs, for short),
a finite set of group elements represented by their encodings.
A black-box group algorithm
is allowed to call the oracles for the group operations and for the
identity test and it might also perform arbitrary bit operations.
The query complexity of an algorithm is the number 
of oracle calls, while the running time or computational complexity 
is the number of oracle calls together with the number of other bit
operations, when we are maximizing over both oracle and standard inputs.
In the quantum setting, the oracles are given by unitary operators.

Many classical algorithms have been 
developed
for computations with 
black-box groups~\cite{BeaBab, BabBea, KanSer}, 
for example the identification of the composition factors, even the non-commutative ones. 
When the oracle operations can be simulated by efficient procedures, 
efficient black-box algorithms automatically produce efficient algorithms. 
Permutation groups, finite matrix groups over finite fields and  
over algebraic number fields fit in this context. 
There has been also considerable effort to design quantum algorithms in black-box groups.
In the case of unique encoding efficient algorithms have been conceived 
for the decomposition of Abelian groups into a direct sum of cyclic groups 
of prime power order~\cite{CheMos},
for
order computing and
membership testing in solvable groups~\cite{Wat}, and 
for
solving  the hidden subgroup problem in
Abelian groups~\cite{Mos}.

The discrete logarithm problem $\DLOG$, and various Diffie-Hellman
type problems are fundamental tasks in computational number
theory. They are equally important in cryptography, since the security
of many cryptographic systems is based on their computational
difficulty.  
Let $G$ be a cyclic group (denoted multiplicatively). Given two
group elements $g$ and $h$, where $g$ is a generator, $\DLOG$ asks to
compute an integer $d$ such that $h=g^d$. Given three
group elements
$g,g^a$ and $g^b$, where $g$ is a generator, the computational
Diffie-Hellman problem $\CDH$ is to compute $g^{ab}$.  Given 
four
group elements $g,g^a, g^b$ and $g^c$, where $g$ is a generator, the
decisional Diffie-Hellman problem $\DDH$ is to decide whether $c=ab$
modulo the order of $g$.

The problems are enumerated in decreasing order of difficulty: $\DDH$
can not be harder than $\CDH$ and $\CDH$ is not harder than
$\DLOG$. While there are groups where even $\DLOG$ is easy (for
example $\Z_m$, the additive group of integers modulo $m$), in general
all three problems are thought to be computationally
hard. 
We are not aware of any group where $\CDH$ is easy while $\DLOG$ is
hard. In fact, Maurer and Wolf have proven in~\cite{MaurerWo99} that
under a seemingly reasonable number-theoretic assumption, the two
problems are equivalent in the case of unique-encoding groups.  Based
on this, Joux and Nguyen~\cite{JouNgu} have constructed a
cryptographic group where $\DDH$ is easy to compute while $\CDH$ is as
hard as $\DLOG$.  In generic black-box groups we have provable query
lower bounds for these problems, even in the case of unique
encoding. Shoup has proven~\cite{Shoup} that in $\Z_p$, given as a
black-box group with unique encoding, to solve $\DLOG$ and $\CDH$
require $\Omega(p^{1/2})$ group operations. Subsequently, Damg{\aa}rd,
Hazay and Zottare~\cite{Dam} have established a lower bound of the
same order for $\DDH$. We remark that the Pohlig-Hellman~\cite{Pohlig}
algorithm reduces $\DLOG$ in arbitrary cyclic groups to $\DLOG$ in its
prime order subgroups. Furthermore, in prime order groups (with unique
encodings), Shanks's baby-step giant-step algorithm 
solves 
the problem
using $O(p^{1/2})$ group 
operations, 
thus matching the lower bound for
black-box groups.

Though, as we said, all three problems are considered computationally
intractable on a classical computer, there is a polynomial time quantum algorithm for $\DLOG$
due to Shor~\cite{Shor}. Since $\DLOG$ is an instance of the Abelian hidden subgroup
problem, Mosca's result~\cite{Mos} implies that by a quantum computer it can also be solved efficiently
in black-box groups with unique encoding.

We are concerned here with identity black-box 
groups,
a special class of black-box groups where only the identity test is given by an oracle.
These groups are quotient groups of some explicitly given ambient group. 
An identity black box group $G$ is specified by an ambient group $G'$ and 
an identity oracle $\Id$ which tests membership in some (unknown) normal subgroup $H$ of $G'$.
In $G$ the group operations are explicitly defined by the group operations in $G'$,
and therefore it is  the quotient group $G'/H$. 

Let $p$ be a prime number. More specifically we will study the problems $\DLOG$, $\CDH$ and $\DDH$ 
in identity black-box groups whose ambient group is $\Z_p^{t+1}$, for some positive integer $t$, and where the
normal subgroup $H$, specified by the identity oracle, is isomorphic to $\Z_p^t$. We denote such an
identity black box group by $G_{p,t}$. We fully characterize the complexity of the three problems in these groups.
Our results are mainly query lower bounds: the probabilistic query complexity of all these problems,
except $\DDH$ in level 1 groups, is $\Omega(p)$, and their quantum query complexity is $\Omega(\sqrt{p})$.
These lower bounds are obviously tight since $\DLOG(G_{p,t})$ can be solved, for all $t \geq 1$,
by exhaustive search and by
Grover's algorithm in respective query complexity $p$ and $O(\sqrt{p})$.
We have also one, maybe surprising, algorithmic result: the computational complexity of $\DDH(G_{p,1})$
is polynomial.  Our results can be summarized in the following theorems.\\ \\
{\bf Theorem (Lower bounds)}
\begin{enumerate}
\item
For all $t \geq 1$, the randomized query complexity of $\DLOG(G_{p,t})$ and $\CDH(G_{p,1})$ is $\Omega(p)$.
\item
For all $t \geq 1$, the quantum query complexity of $\DLOG(G_{p,t})$ and $\CDH(G_{p,1})$ is $\Omega(\sqrt{p})$.
\item
For all $t \geq 2$, the randomized query complexity of $\DDH(G_{p,t})$ is $\Omega({p})$.
\item
For all $t \geq 2$, the quantum query complexity of $\DDH(G_{p,t})$ is $\Omega(\sqrt{p})$.
\end{enumerate}
{\bf Theorem (Upper bound)}
$\DDH(G_{p,1})$ can be solved in probabilistic polynomial time in $\log p$.

\section{Preliminaries}
Formally, a {\em black-box group} $G$ is a 4-tuple 
$G = (C, \Mult, \Inv, \Id)$ where $C$ is the
set of admissible codewords, 
$\Mult:C\times C\mapsto C$ is a binary operation, $\Inv:C\rightarrow C$ is a unary operation
and $\Id:C\rightarrow \{0,1\}$ is a unary Boolean function. The operations $\Mult, \Inv$ and the function
$\Id$ are given by {\em  oracles}.
We require that 
there exists a finite group $\widetilde G$ and a surjective 
map $\phi : C \rightarrow \widetilde G$ for which, for every $x,y \in C$, we have $\phi(\Mult(x,y))=
\phi(x)\phi(y)$, $\phi(\Inv(x))=\phi(x)^{-1}$, and $\phi(x)=1_{\widetilde G}$ if and only
if $\Id(x)=1$.
We say that $x$ is (more accurately, encodes)
the identity element in $G$ or  that
$x=1$ if $\Id(x)=1$. With the identity oracle we can test equality since 
$x=y$ in $G$ exactly when $\Id(\Mult(x,\Inv(y)))=1$. We say that a black-box group has {\em unique
encoding} if $\phi$ is a bijection. For abelian groups we also use the additive notation in which case
the binary operation of $G$ is denoted by $\Add$ and its identity element is denoted by $0$.

We are concerned here with a special class of black-box groups which are quotient groups of some explicitly
given group. 
An {\em identity black-box group} is a couple $G = (G', \Id)$ where $G'$ is group and 
the {\em identity oracle} $ \Id:G'\rightarrow \{0,1\}$
is the characteristic function of some (unknown) normal subgroup $H$ of $G'$. 
We call $G'$ the {\em ambient group} of $G$.
In  $G$ the group operations $\Mult$ and $\Inv$ are defined by the group operations in the ambient group $G'$
modulo $H$.
As a consequence, $G$ is the quotient group $G'/H$. 

Let $p$ be a prime number and  let $t\geq 1$ be a positive integer. 
We denote by $\Z_p$ the additive group of integers modulo $p$, by $\F_p$  the finite field of size $p$, 
and by $\Z_p^t$ the $t$-dimensional vector space over $\F_p$.
For $h,k \in \Z_p^{t}$, we denote their scalar product modulo $p$ by $h \cdot k$.

We will work with identity black-box groups whose ambient group is $\Z_p^{t+1}$, and 
the subgroup $H$ is isomorphic to $\Z_p^t$, that is a hyperplane of   $\Z_p^{t+1}$.
Regarding the problems we are concerned with,
the only real restriction of this model is that our black-box group $G$
has (known) prime order $p$. Indeed, let 
$G = (C, \Mult, \Inv, \Id)$ be a cyclic black box group
of order $p$. Let $g_1=g$, $g_2=g^a$, $g_3=g^b$ and $g_4=g^c$ be the
input quadruple for $\DDH$. We define maps 
$\psi_i:\{0,\ldots, p-1\}\rightarrow C$ as $\psi(x)=g_i^x$. Here 
$g_i^x$ is computed by a fixed method based on repeated squaring
and the binary expansion of $x$ using $\poly(\log p)$ iterated applications of
the oracle $\Mult$. We also define an oracle $\Id'$ on 
$\{0,\ldots,p-1\}^4$ by $\Id'(x_1,x_2,x_3,x_4)=
\Id(\phi_1(x_1),\phi_2(x_2),\phi_3(x_3),\phi_4(x_4))$. It is not
difficult to see that this makes $G$ an identity black box group with 
ambient group $\Z_p^4$ where the identity oracle $\Id'$ can be implemented
using a $\poly(\log p)$ calls to $\Mult$ and a single call to $\Id$.
This recipe reduces the given instance of $\DDH$ to an instance
of $\DDH$ in the new setting in an obvious way. Similarly,
$\DLOG$ and $\CDH$ can be reduced to instances with
ambient groups $Z_p^2$ and $\Z_p^3$, respectively.

We will specify the identity oracle by a non-zero normal vector $n \in \Z_p^{t+1}$ of $H$.
By permuting coordinates
and multiplying by some non-zero constant, we can suppose without loss of generality that it is of the form
$n = (1, n_1, \ldots, n_t)$. We call such a vector {\em $t$-suitable}. We define the 
function $\Id_n : \Z_p^{t+1} \rightarrow \{0,1\}$ by $\Id_n (h) = 1$ if $h \cdot n = 0$.
Clearly $\Id_n$ is the characteristic function of the hyperplane $H_n = \{h  \in \Z_p^{t+1}: h \cdot n = 0\}.$
We define the identity black-box group $G_{p,t} = (\Z_p^{t+1}, \Id)$, where the identity oracle $\Id$ satisfies
$\Id = \Id_n$, for some (unknown) $t$-suitable vector $n$.
We call $t$ the {\em level} of the group $G_{p,t}.$
We emphasize again that the group operations of $G_{p,t}$ are performed as group operations in $\Z_p^{t+1}$.
Therefore, for $h,k \in \Z_p^{t+1}$,  the equality $h=k$ in  $G_{p,t}$
means equality in $\Z_p^{t+1}$ modulo $H_n$, where $H_n$ is identified by $\Id_n$.
To be short, we will refer to $G_{p,t}$ as the {\em hidden cyclic group} of level $t$.
We remark that any lower bound for $t$-suitable $n$ remain 
trivially valid for general normal vector $n$. Also, 
the first nonzero coordinate of $n$ can be found 
using at most $t$ queries (namely $\Id(1,0,0,\ldots,0)$,
$\Id(0,1,0, \ldots,0)$, $\ldots$, $\Id(0,\ldots,0,1,0)$). Furthermore,
scaling this coordinate to $1$ does not affect the oracle $\Id_n$. Therefore
$t$-suitability of $n$ affects any upper bound by at most $t$ queries.

\begin{proposition}
\label{prop:group}
The groups $G_{p,t}$ 
and $\Z_p$ are isomorphic and the map $\phi : G_{p,t} \rightarrow \Z_p$ defined by
$\phi(h) = h \cdot n \in \Z_p$ is a group isomorphism.

\end{proposition}

\begin{proof}
The maps from $\Z_p^{t+1}$ to $G_{p,t}$ (respectively to $\Z_p$) mapping
$h \in \Z_p^{t+1}$ to  its class in the quotient $G_{p,t}$
(respectively to $h \cdot n$) are group homomorphisms with the same kernel $H_n$.
\end{proof}

We recall now the basic notions of query complexity for the specific case of Boolean functional oracle problems.
Let 
$m$ be a positive integer. 
A {\em functional oracle problem} is a function $A: S \rightarrow \{0,1\}^M$, 
where $S \subseteq \{0,1\}^m$ and $M \geq 1$ is a positive integer.
If $M=1$, then we call the functional oracle problem {\em Boolean}.
The input $f \in S$ is given by an oracle, that is $f(x)$ can be accessed by the query $x$. The output on
$f$ is $A(f)$. 
Each query adds one to the complexity of an algorithm, but
all other computations are free.
The state of the computation is represented by three
registers, the query register $1 \leq x \leq m$, the answer register $a \in \{0,1\}$, and the 
work register $z$. The computation takes place in the vector space spanned by all
basis states $|x\rangle|a\rangle|z\rangle$.
In the {\em quantum query} model introduced by Beals et al.~\cite{BBC}
the state of the computation is a complex
combination of all basis states which has unit length in the norm $l_2$.
In the {\em randomized query } model it is a non-negative real combination of unit length
in the norm $l_1$, and in the deterministic model it is always one of the basis states.

The query operation $O_f$ maps the 
basis state
$|x\rangle|a\rangle|z\rangle$ 
into the state $|x\rangle|(a+f(x)) \bmod 2 \rangle|z\rangle$.
Non-query operations do not depend on $f$.
A {\em $k$-query algorithm} is a sequence of $(k+1)$ operations
$(U_0, U_1, \ldots , U_k)$ where $U_i$ is unitary in the quantum,
and stochastic in the randomized model. 
Initially the state of the computation is set to some
fixed value $|0\rangle|0\rangle|0\rangle$, and then the sequence of operations
$U_0, O_f, U_1, O_f$, $\ldots$, $U_{k-1}$, $O_f$, $U_k$ is applied.
A quantum or randomized
algorithm {\em computes} $A$ on input $f$
if the observation of the last $M$ bits of the work register
yields $A(f)$ with probability 
at least $2/3$.
Then $\QC(A)$
(respectively $\RC(A)$) is the smallest $k$ for which there exists
a $k$-query quantum (respectively randomized) algorithm which 
computes $A$ on every input $f$. 
We have 
$ \RC(A) \leq \QC(A) \leq m$.

We define now the
problems we are concerned with,
the {\em discrete logarithm} problem $\DLOG$, the {\em computational Diffie-Hellman} problem $\CDH$ 
and the {\em decisional Diffie-Hellman} problem $\DDH$ in 
hidden cyclic  groups $G_{p,t}$. 
As in the rest of the paper the additive notation 
will to be more convenient, in contrast to
the informal definitions if the introduction, we use here 
the additive terminology.
We say that a 
quadruple $(g,h,k, \ell) \in G_{p,t}^4$ 
is a DH-{\em quadruple} if
$g$ is a generator of $G_{p,t}$,
$h = ag, k = bg$ and $\ell = cg$ for some integers $a,b,c$ such that $c = ab$ modulo $p$.
\\
\vbox{\begin{quote}
\DLOG$(G_{p,t})$ \\
{\em Oracle input:} $\Id_n$ for some $t$-suitable vector $n$. \\
{\em Input:} 
A couple $(g,h) \in G_{p,t}^2$ such that $g$ is a generator of $G_{p,t}$.\\
{\em Output:} A non-negative integer $d$ such that $dg =h$.
\end{quote}}\\
\vbox{\begin{quote}
\CDH$(G_{p,t})$ \\
{\em Oracle input:} $\Id_n$ for some $t$-suitable vector $n$. \\
{\em Input:} 
A triple $(g,h,k) \in G_{p,t}^3$ such that $g$ is a generator of $G_{p,t}$.\\
{\em Output:} $\ell \in G_{p,t}$ such that $(g,h,k, \ell) $  is a DH-quadruple.
\end{quote}}\\
\vbox{\begin{quote}
\DDH$(G_{p,t})$ \\
{\em Oracle input:} $\Id_n$ for some $t$-suitable vector $n$. \\
{\em Input:} 
A quadruple $(g,h,k, \ell) \in G_{p,t}^4$ such that $g$ is a generator of $G_{p,t}$.\\
{\em Question:} Is $(g,h,k, \ell)$ a DH-quadruple?
\end{quote}}\\

An algorithm for these problems has access to the input
and oracle access to the oracle input,
and every query is counted as one computational step.
We say that it solves the problem {\em efficiently} 
if it works in time polynomial in $\log p$ and $t$.
For any fixed input, the problems become functional oracle problems,
where we consider only those identity oracles for which the input is legitimate.
By their {\em query complexity} we mean,
both in the quantum and in the randomized model,
the maximum, over all inputs, of the respective query complexity of these 
functional oracle problems.

The problems are enumerated in decreasing order of difficulty.
The existence of an efficient algorithm for \DLOG$(G_{p,t})$ implies the existence
of an efficient algorithm for \CDH$(G_{p,t})$, which in turn gives rise to an efficient algorithm 
for \DDH$(G_{p,t})$. For query complexity we have 
$\QC(\DDH(G_{p,t})) \leq \QC(\CDH(G_{p,t}))+1$
and
$\QC(\CDH(G_{p,t})) \leq 2\QC(\DLOG(G_{p,t}))$, 
and the same inequalities hold for the randomized model.
The problems are getting harder as the level
of the hidden cyclic group increases, as 
the almost trivial reductions in the next Proposition show. To ease notation,
for $h = (h_0, \ldots, h_t) $ in $\Z_p^{t+1}$,
we denote by $h'$ the element $(h_0, \ldots, h_t, 0) \in \Z_p^{t+2}$.

\begin{proposition}
\label{prop:reduction} 
For every $t \geq 1$,
$\DLOG(G_{p,t})$ and $\DDH(G_{p,t})$ are polynomial time 
many-one reducible to  respectively $\DLOG(G_{p,t+1})$ and  $\DDH(G_{p,t+1})$;
and $ \CDH(G_{p,t})$ is commutable in polynomial time with a single query to $\CDH(G_{p,t+1})$.
\end{proposition}

\begin{proof}
First observe that the identity oracle 
$\Id_{n'}$ in $G_{p,t+1}$ can be simulated 
by the identity oracle $\Id_n$ of $G_{p,t}$. Indeed, for an arbitrary element $h^*$  in $G_{p,t+1}$, where
$h^* = (h_0, h_1, \ldots , h_t, h_{t+1})$, set $h = (h_0, h_1, \ldots , h_t)$. 
Then $h^* \cdot n' = h \cdot n$.
Let $g$ be a generator of $G_{p,t}$ with identity oracle $\Id_n$, that is $g \cdot n \neq 0$. 
Then $g' = (g,0)$ is a generator of
$G_{p,t+1}$ with identity oracle $\Id_{n'}$, 
since $g' \cdot n' = g \cdot n$, and therefore $g' \cdot n' \neq 0$. 

For arbitrary
$g,h,k,\ell \in G_{p,t}$, and for every integer $d$, we have
$dg =h$ if and only if $dg' =h'$. Similarly, 
$(g,h,k, \ell) $  is a DH-quadruple if and only if 
$(g',h',k', \ell') $  is a DH-quadruple. This gives the many-one reductions for $\DLOG$ and $\DDH$.
In the $\CDH$ reduction, on an instance $(g,h,k)$, 
we call $\CDH(G_{p,t+1})$ on instance $(g',h',k')$. Suppose that it gives the answer
$\ell^* =  (\ell_0, \ell_1, \ldots , \ell_t, \ell_{t+1})$. We set 
$\ell =  (\ell_0, \ell_1, \ldots , \ell_t)$, 
and observe that  $(g,h,k, \ell) $  is a DH-quadruple because 
$(g',h',k', \ell^*) $  is a DH-quadruple and $\ell^* \cdot n' = \ell \cdot n$.
\end{proof}


\section{The complexity in groups of level 1}

In most parts of this section we restrict ourselves to the case $t=1$. To simplify notation, we set $n = (1,s)$ 
and we denote 
the identity oracle $\Id_n$ by $\Id_s$ and the line $H_n$ 
of the plane $\Z_p^2$
by $H_s$. 
Also, we refer to $s$ as the {\em secret}. As it turns out, solving $\DLOG(G_{p,1})$
or $\CDH(G_{p,1})$ is essentially as hard as finding the secret, therefore 
we 
formally define 
this problem as
\\
\vbox{\begin{quote}
\SECRET$(G_{p,1})$ \\
{\em Oracle input:} $\Id_s$ for some $s \in \Z_p$.
\\
{\em Output:} $s$.
\end{quote}}

What is the 
query
complexity of finding $s$, that is how many calls to the identity oracle are needed for that task?
To answer this question, we consider $\US$, the well studied unstructured search problem.
For this, let $C$ be an arbitrary set, and let $s \in C$ be an
arbitrary distinguished element. Then the {\em Grover oracle} $\Delta_s : C \rightarrow \{0,1\}$
is the Boolean function such that $\Delta_s(x) = 1$ if and only if $x=s$. 
The {\em unstructured search} problem over $C$ is defined as
\\
\vbox{\begin{quote}
\US$(C)$ \\
{\em Oracle input:} $\Delta_s$ for some $s \in C$.
\\
{\em Output:} $s$.
\end{quote}}

Suppose that the size of $C$ is $N$. 
It is easily seen that  probabilistic query complexity of 
$\US(C)$ is linear in $N$. The quantum query complexity of the problem is also well studied.
Grover~\cite{Grover} has determined that it can be solved with $O(\sqrt{N})$ queries, while Bennett 
et 
al.~\cite{BBBV}
have shown that $\Omega(\sqrt{N})$ queries are also necessary.

\begin{fact}
\label{fact}
For $|C| = N $, the randomized query complexity of $\US(C)$ is $\Theta (N)$ and its quantum query complexity is
$\Theta(\sqrt{N})$.
\end{fact}

The relationship between $\US$ and the problem $\SECRET$ is given by the fact that
the identity oracle $\Id_s$ and the Grover oracle $\Delta_{s}$ 
can simulate each other with a single query.

\begin{proposition}
\label{prop:simulation}
The identity oracle $\Id_s$ of  $G_{p,s}$ and the
Grover oracle $\Delta_{s}$, defined over $\Z_p$,
can simulate each other with  at most one query.
\end{proposition}
\begin{proof}
The simulation of the Grover oracle by the identity oracle is simple: for $ x \in \Z_p$ just query $\Id_s$ on
$(x,-1)$. 

For the reverse direction, 
let $h = (h_0,h_1)$ be an input to the identity oracle. Then
$h$ 
encodes
the identity element, that is 
$\Id_s(h) = 1$,
 if and only if $-h_0 =sh_1$. When $h_1$ is invertible in $\Z_p$
we can check by the Grover oracle if $-h_0 h_1^{-1} =s$. For $h_1 = 0$ the only possible value 
for $h_0$
to put $h$ into $H_s$
is 0. Therefore we have
\[
  \Id(h) =
  \begin{cases}
    1 & \text{if $ h = (0,0)$} \\
    0 & \text{if $ h_1 = 0 $ and $h_0 \neq 0$} \\
    \Delta_s(-h_0 h_1^{-1}) & \text{otherwise}.
  \end{cases}
\]
\end{proof}

\begin{corollary}
\label{cor:secret}
The randomized query complexity of \SECRET$(G_{p,1})$ is $\Theta (p)$ and its quantum query complexity is
$\Theta(\sqrt{p})$.
\end{corollary}

%

We will now consider the reductions of \SECRET$(G_{p,1})$ to $\DLOG(G_{p,1})$ and $\CDH(G_{p,1})$.
The case of $\DLOG(G_{p,1})$ in fact follows from the case of $\CDH(G_{p,1})$, but it is so simple
that it is worth to describe it explicitly.

\begin{lemma}
\label{lem:DLOG}
The secret $s$ in $G_{p,1}$ can be found with a single oracle call
to $\DLOG(G_{p,1})$
\end{lemma}
\begin{proof}
First observe that $(1,0)$ is a generator of $G_{p,1}$, for every $s$.
The algorithm calls $\DLOG(G_{p,1})$
on input $(g,h) = ((1,0), (0,1))$. Since $\phi(g)  =1$ and $\phi(h)  =s$
where $\phi$ is as in 
Proposition~\ref{prop:group},
the oracle's answer is
the secret $s$ itself.
\end{proof}

We remark that with overwhelming probability we could have given also a random couple $(g,h) \in G_{p,1}^2$ to the oracle, where $g$ is a generator. Indeed, let's suppose that 
$d$ is the discrete logarithm. Then $h - dg \in H_s$, and therefore $s = - (h_0 -dg_0) (h_1 - d g_1)^{-1}$, where the
operations are done in $\Z_p$, under the condition that $h_1 - d g_1 \neq 0$, which happens with probability
$(p-1)/p$.

The reduction of \SECRET$(G_{p,1})$ to  $\CDH(G_{p,1})$ requires more work. The main idea of the reduction
is to extend 
$G_{p,t}$
to a field and use the multiplication for the characterization of DH-quadruples.
Indeed, since $\Z_p$ is the additive group of the field $\F_p$, we can use the isomorphism $\phi$ 
of Proposition~\ref{prop:group} between $G_{p,t}$ and $\Z_p$
to define appropriate multiplication  and multiplicative inverse operations. This extends
$G_{p,t}$ to a field isomorphic to $\F_p$ which we denote 
by $F_{p,t}$. This process is completely standard but we describe it for completeness.
The definitions of these two operations are:
\begin{align*}
hk  & = \phi^{-1} (\phi(h) \phi(k)),  \\
    h^{-1} & = \phi^{-1}(\phi(h)^{-1}).
\end{align*}
With these operations the map $\phi$ becomes a field isomorphism between $F_{p,t}$ and $\F_p$.

\begin{proposition}
\label{prop:field}
The map $\phi$ of Proposition~$\ref{prop:group}$ is an isomorphism between $F_{p,t}$ and $\F_p$.
\end{proposition}
\begin{proof}
By definition $\phi(hk) = \phi(h) \phi(k)$ and   $\phi ( h^{-1} ) = \phi(h)^{-1}.$
\end{proof}

The field structure of $F_{p,t}$ yields a very useful characterization of DH-quadruples. 
\begin{proposition}
\label{prop:DH-quadruple}
Let $g$ be a generator of $G_{p,t}$. In $F_{p,t}$ the quadruple $(g,h,k,\ell)$ is a $\fuc$-quadruple if and only if
$$
g\ell - hk = 0.
$$
\end{proposition}
\begin{proof}
Let $h = ag, k = bg$ and $\ell = cg$ for some integers $a,b,c$. Using the field structure of $F_{p,t}$,
it is true that $g\ell - hk = 0$ if and only if $(c-ab)g^2 = 0$. Since $F_{p,t}$ is isomorphic to $\F_p$, an element
$g$ is a generator of the additive group $\Z_p$ exactly when $g \neq 0$, and therefore when $g^2$ is a generator.
Therefore $(c-ab)g^2 = 0$ if an only if $c = ab$.
\end{proof}

We define the application $\chi : \Z_p^{t+1} \rightarrow \F_p[x_1, \ldots , x_t]$, 
from  $\Z_p^{t+1}$ to the ring of $t$-variate polynomials
over $\F_p$, where the image $\chi(h)$ of $h = (h_0, h_1, \ldots, h_t) \in \Z_p^{t+1}$ is the polynomial 
$p_h(x_1, \ldots , x_t) = 
h_0 + \sum_{i=1}^t h_ix_i$.
Observe that $p_h(n_1, \ldots n_t) = h \cdot n$, 
therefore the isomorphism $\phi$ between $G_{p,t}$ with identity oracle $\Id_n$ and $\Z_p$ can 
also be
expressed as
$\phi(h) = p_h(n_1, \ldots n_t)$.

\begin{proposition}
\label{prop:characterization}
Let $g$ be a generator of $G_{p,1}$ and let $h,k,\ell$ be arbitrary elements.
Then $(g,h,k,\ell)$  is a DH-quadruple if and only if $s$ is a root
of the polynomial $p_g(x) p_{\ell}(x) - p_h(x) p_k(x)$.
\end{proposition}
\begin{proof}
By Proposition~\ref{prop:DH-quadruple} we know that $(g,h,k,\ell)$  is a DH-quadruple if and only if
$g\ell - hk = 0$, that is when $p_{g\ell - hk}(s) =0$. Now 
Proposition~\ref{prop:field} implies that this happens exactly when $p_g(s) p_{\ell}(s) - p_h(s) p_k(s) = 0.$

\end{proof}

\begin{lemma}
\label{lem:CDH}
There is a probabilistic polynomial time algorithm which, given oracle access
to $\CDH(G_{p,1})$, solves \SECRET$(G_{p,1})$. 
The algorithm asks a single query to $\CDH(G_{p,1})$. 
If we are also given a quadratic non-residue in $\Z_p$,
the algorithm can be made deterministic.
\end{lemma}
\begin{proof}
The algorithm sets $g = (1,0), h= (0,1), k=(1,1)$ and presents it to the oracle. Let the oracle's answer be
$\ell = (\ell_0, \ell_1)$. Since $(g,h,k, \ell)$ is a DH-quadruple, by Proposition~\ref{prop:characterization}
we have that $s$ is the root of the second degree equation
$$
x^2 + (1-\ell_1)x + \ell_0 = 0.
$$
Assuming
that a quadratic non-residue in $\Z_p$ is available then the (not necessarily distinct) roots
$x_1, x_2$ can be computed in deterministic polynomial time
using the Shanks-Tonelli algorithm~\cite{Shanks}.
Without this assumption, a quadratic non-residue can always be computed in probabilistic
polynomial time
because for $p>2$ the quadratic residues form a subgroup
of index two of the multiplicative group of $\F_p$ and hence
$p>2$ half of the nonzero elements in $\Z_p$ are not squares.
Finally, we make at most two calls to $\Id_s$ 
on $(x_1, -1)$ and on $(x_2, -1)$. The positive answer tells us which one 
of the roots is the secret $s$.
\end{proof}

Similarly to the $\DLOG$ case, we could have presented 
with overwhelming probability also a random triple $(g,h,k) \in G_{p,s}^3$ to $\CDH(G_{p,1})$, 
where $g$ is a generator. 
Indeed, if the oracle answer is $\ell = (\ell_0, \ell_1)$ then $s$ is a root of the 
(at most second degree)
 equation
$$
(g_0 + g_1x)(\ell_0 + \ell_1 x) = (h_0 + h_1 x) (k_0 + k_1x).
$$
If the equation is of degree 2 then we can proceed as in the proof of Lemma~\ref{lem:CDH}.
This happens exactly when $h_1k_1 \neq g_1 \ell_1$. But for every possible fixed value $a$ for $g_1 \ell_1$,
the probability, over random $h_1$ and $k_1$, that $h_1k_1 = a$ is at most $2/p$, the worst case being $a=0$.
Therefore a random triple $(g,h,k)$ would be suitable for the proof with probability at least $(p-2)/p.$

\begin{theorem}
\label{thm:DLOG}
The following lower bounds hold for the query complexity of $\DLOG$ and $\CDH$: 
\begin{enumerate}
\item[$(1)$]
The classical query complexity of both $\DLOG(G_{p,s})$ and $\CDH(G_{p,s})$ is $\Omega({p})$.
\item[$(2)$]
The quantum query complexity of both $\DLOG(G_{p,s})$ and $\CDH(G_{p,s})$ is $\Omega(\sqrt{p})$.
\end{enumerate}
\end{theorem}

\begin{proof}
Let us suppose that with $m$ queries to the identity oracle $\Id_s$, one can solve 
$\DLOG(G_{p,1})$ or $\CDH(G_{p,1})$. Respectively 
Lemma~\ref{lem:DLOG} 
and Lemma~\ref{lem:CDH} imply that \SECRET$(G_{p,1})$ can be solved with $m$ queries.
The result then follows from the lower bounds of  Corollary~\ref{cor:secret}.
\end{proof}

\begin{theorem}
\label{thm:DDH}
The $\DDH(G_{p,1})$ problem can be solved in probabilistic polynomial time. 
If we are given a quadratic non-residue in $\Z_p$ 
the algorithm can be made deterministic.
\end{theorem}

\begin{proof}
Let $(g,h,k, \ell)$ be an input to $\DDH(G_{p,1})$ where $g$ is a generator of $G_{p,1}$.
By Proposition~\ref{prop:characterization} it is a DH-quadruple if and only if $s$ is a root
of the polynomial $p_g(x) p_{\ell}(x) - p_h(x) p_k(x)$, and that is what the algorithm checks.
When the polynomial is constant, then the answer is yes if the constant is zero, and otherwise it is no.
When the polynomial is non constant, the algorithm essentially proceeds as the one in
Lemma~\ref{lem:CDH}. It solves the 
(at most second degree) 
equation and then checks with the identity oracle
if one root is equal to $s$. 
\end{proof}

\section{The complexity of $\DDH$ in groups of level 2}
There are several powerful means to prove quantum query lower bounds, 
most notably the adversary and the polynomial method~\cite{BBC}.
The quantum adversary method initiated by Ambainis \cite{Amb} has
been extended in several ways. The most powerful of those, the method using negative weights~\cite{HLS},
turned out to be an exact characterization of the quantum query complexity~\cite{LMR}.
We use here a special case of the positive weighted adversary 
method~\cite{Aaronson_min_search,Ambainis2,Zha2} that also gives probabilistic
lower bounds~\cite{Aaronson_min_search}. 


\begin{fact}
\label{fact:weights}
Let $A: S \rightarrow \{0,1\}$ be a  Boolean functional oracle problem, where $S \subseteq \{0,1\}^m$.
For any $S \times S$ matrix $M$, set 
$$\sigma(M,f)=\sum_{g\in S} M[f,g].$$
Let $\Gamma$ be an arbitrary $S\times S$ nonnegative symmetric matrix
that satisfies $\Gamma[f,g]=0$ whenever $A(f)=A(g)$. For $1\leq x \leq m$, let $\Gamma_x$ be the matrix
$$
\Gamma_x[f,g]=
\begin{cases}
0 & \text{ if } f(x)=g(x), \\
\Gamma[f,g] & \text{ otherwise.}
\end{cases}
$$
Then
$$ \QC(A)=\Omega\left(\min_{\Gamma[f,g]\neq 0,f(x)\neq g(x)}\sqrt{\frac{\sigma(\Gamma,f)\sigma(\Gamma,g)}
  {\sigma(\Gamma_x,f)\sigma(\Gamma_x,g)}}\right),$$
$$ \RC(A)=\Omega\left(\min_{\Gamma[f,g]\neq 0,f(x)\neq g(x)}\max\left\{\frac{\sigma(\Gamma,f)}{\sigma(\Gamma_x,f)},
   \frac{\sigma(\Gamma,g)}{\sigma(\Gamma_x,g)}\right\}\right). $$
\end{fact}

\begin{theorem}
\label{thm:level2}
The following lower bounds hold for the query complexity of $\DDH$ in level $2$ 
hidden cyclic groups:
$$
\QC(\DDH(G_{p,2})) = \Omega(\sqrt{p}) ~~{\rm and} ~~
\RC(\DDH(G_{p,2})) = \Omega({p}).
$$
\end{theorem}

\begin{proof}
Let $i= ((1,0,0), (0,1,0), (0,0,1), (0,1,1))$.
Observe that the element $(1,0,0)$ is a generator
of $G_{p,2}$, for any 2-suitable vector $n=(1, n_1, n_2)$. 
By Proposition~\ref{prop:DH-quadruple}, we know that $i$ is a DH-quadruple if and only if
$n_1 + n_2 = n_1n_2$. We say that $n$ is {\em positive} if this equality holds, otherwise we say that it is
{\em negative}.
Let $m = p^3$ and let $S = \{ \Id_n : n \in \Z_p^2\}$.
We will apply Fact~\ref{fact:weights} to the Boolean functional oracle problem $\DDH$
defined in $G_{p,2}$ on input $i$ with the oracle input being the identity oracle
$\Id_n : \Z_p^3 \rightarrow \{0,1\}$. For simplicity we will refer to this Boolean functional oracle problem
just by $\DDH(n)$.
We define the symmetric $p^2 \times p^2$ Boolean adversary matrix $\Gamma$ as follows:
$$
\Gamma[n,n']=
\begin{cases}
1 & \text{ if } \DDH(n) \neq \DDH(n'), \\
0 & \text{ otherwise,}
\end{cases}
$$
where again $\Gamma[n,n']$ is a shorthand notation for $\Gamma[\Id_n,\Id_{n'}].$

We first determine $\sigma(\Gamma, n)$. If $n_1 = 1$ then there is no $n_2$ such that $n_1 + n_2 = n_1n_2$.
Otherwise, for every fixed $n_1 \neq 1$, there is a unique $n_2$ that makes this equality hold, in particular
$n_2 = n_1 (n_1 -1)^{-1}$. Therefore the number of positive $n$ is $p-1$ and the number of negative $n$ is
$p^2 -p +1$. Thus we have the following values for $\sigma(\Gamma, n)$:
$$
\sigma(\Gamma, n)=
\begin{cases}
p^2 -p +1 & \text{ if } n \text{ is positive, }\\
p-1 & \text{ otherwise.}
\end{cases}
$$
 
Let us recall, that by definition, for every $h \in G_{p,2}$,

\begin{equation}
\label{lower}
\Gamma_h[n,n']=
\begin{cases}
1 & \text{ if } \DDH(n) \neq \DDH(n') \text{ and } \Id_n(h) \neq \Id_{n'}(h), \\
0 & \text{ otherwise.}
\end{cases}
\end{equation}

We fix now $n$ and $n'$ such that $\DDH(n) \neq \DDH(n')$, we will suppose without loss of generality that
$n$ is positive and $n'$ is negative.
We also fix $h = (h_0, h_1, h_2)$ in $\Z_p^3$
such that $\Id_n(h) \neq \Id_{n'}(h)$. This implies that $(h_1, h_2) \neq (0,0)$.
We want to lower bound  $\sigma(\Gamma, n)/ \sigma(\Gamma_h,n)$ and
$\sigma(\Gamma, n')/ \sigma(\Gamma_h,n')$. Obviously both fractions are at least 1.
We distinguish two cases according to whether $\Id_n(h) = 0$ or $\Id_{n'}(h) = 0$.

Case 1: $\Id_{n'}(h) = 0$. Then 
$$
\sigma(\Gamma_h,n') = \left |\left \{(m_1,m_2) \in \Z_p^2 : 
\begin{array}{l}
m_1 + m_2 = m_1 m_2 
\text{ and } \\
h_0 + h_1m_1 + h_2m_2 = 0
\end{array}
\right\}\right|.
$$
We claim that the carnality at the right hand side is at most 2. We know already that $m_1 \neq 1$ and
$m_2 = m_1 (m_1 -1)^{-1}$. Therefore $m_1$ satisfies the second degree equation 
$$
h_1x^2 + (h_0 -h_1 +h_2)x -h_0 = 0.
$$
The number of roots of this equation is at most 2, unless the polynomial is 0.
But this can not be the case, because then $h_1=h_2 = 0$, a contradiction.
Therefore, taking into account~(\ref{lower}), we have
\begin{equation}
\label{neg}
\frac{\sigma(\Gamma, n')}{ \sigma(\Gamma_h,n')} = \Omega \left(\frac{p}{1}\right) = \Omega (p).
\end{equation}

Case 2: $\Id_{n}(h) = 0$. Then
$$
\sigma(\Gamma_h,n) \leq |\{(m_1,m_2) \in \Z_p^2 :  h_0 + h_1m_1 + h_2m_2 = 0\}|.
$$
Since $(h_1,h_2) \neq (0,0)$, the number of roots of this linear equation with two variables is $p$.
Therefore, again taking into account~(\ref{lower}), we have
\begin{equation}
\label{pos}
\frac{\sigma(\Gamma, n)}{ \sigma(\Gamma_h,n)} = \Omega \left(\frac{p^2}{p}\right) = \Omega (p).
\end{equation}
The statements of the theorem immediately follow from equations~(\ref{neg}) and~(\ref{pos}).
\end{proof}

Similarly to the remarks after Lemmas~\ref{lem:DLOG} and~\ref{lem:CDH}, we could have used
in the proof instead of $i$ a random input $(g,h,k,\ell)$, with high probability
of success.
Indeed, if we can show that the number of solutions
of the system of equations
\[
  \begin{cases}
    (g_0 + g_1x + g_2y) (\ell_0 + \ell_1 x + \ell_2 y) \\
	\mbox{~~}- (h_0 +h_1x +h_2y) (k_0 +k_1x +k_2 y)  & = 0 \\
    1 +u_1x + u_2 y & =   0
  \end{cases}
\]
is at most 2
for every 
$u = (1, u_1, u_2)$ in $\Z_p^3$, with $(u_1, u_2) \neq (0,0)$, 
then the same proof works.
To see what we claim
we observe first that $g_2\ell_2-h_2k_2$ is nonzero with probability at least
$(p-1)/p$. If this happens then we are done with every $u$
of the form $u=(1,u_1,0)$. Indeed, in that case $u_1\neq 0$ and 
the second equation implies $x=-1/u_1$ and by substituting
this in the first equation we obtain an equation in $y$ with
a proper quadratic term. 
To deal with those $u$ for which $u_2\neq 0$ we set
$\alpha=1/u_2$ and $\beta=u_1/u_2$. By the second equation
we have $y=-\beta x-\alpha$ and substituting this in the first
equation the polynomial becomes
$$P_0+P_1x+P_2x^2,$$
with
$$
P_0=A+B\alpha+C\alpha^2,\;\;
P_1=D+E\alpha+B\beta+2C\alpha\beta\;\;\mbox{and}\;\;
P_2=(F+E\beta+C\beta^2),
$$
where
$A=g_0\ell_0-h_0k_0$,
$B=h_0k_2+h_2k_0-g_0\ell_2-g_2\ell_0$,
$C=g_2\ell_2-h_2k_2)$,
$D=g_0\ell_1+g_1\ell_0-h_0k_1-h_1k_0$,
$E=h_1k_2+h_2k_1-g_1\ell_2-g_2\ell_1$ and
$F=g_1\ell_1+h_1k_1$.
Using Macaulay2~\cite{Macaulay2}, one can show that the ideal of 
$\F_p[g_0,\ldots,\ell_2,\alpha,\beta]$ generated by $P_0,P_1$ and $P_2$
contains a nonzero polynomial $f$ of degree six from $\F_p[g_0,\ldots,\ell_2]$. 
By the Schwartz-Zippel lemma~\cite{Schwartz,Zippel},
 $f$ takes a nonzero value with
probability at least $1-6/p$. If that happens then there exist no
$\alpha,\beta$ making the three coefficients $P_0,P_1$ and $P_2$
simultaneously zero. The overall probability of choosing a good $g,h,k,\ell$     
is therefore at least 1-7/p.

\section*{Acknowledgments}
Parts of this research was accomplished while the first two authors were
visiting the Centre for Quantum Technologies at the National University of Singapore
and also
while the last two authors were
Fellows in November 2018 at the Stellenbosch Institute for Advanced Study.
They would like to thank STIAS for support and hospitality.
The research at CQT was supported by the
Singapore National Research Foundation, the Prime Minister's Office, Singapore and the Ministry of Education, Singapore under the Research Centres of Excellence programme under research grant R 710-000-012-135.
This research was also partially funded by
the  Hungarian National Research, Development and Innovation Office --
NKFIH. In addition, this work has been supported in part by the European Union as H2020
Programme under grant agreement number ERC-669891 
and by the QuantERA ERA-NET Cofund project QuantAlgo.

\bibliography{nonunique}

\begin{thebibliography}{10}

\bibitem{Aaronson_min_search}
Scott Aaronson.
\newblock Lower bounds for local search by quantum arguments.
\newblock {\em SIAM J. Comput.}, 35(4):804--824, 2006.

\bibitem{Amb}
Andris Ambainis.
\newblock Quantum lower bounds by quantum arguments.
\newblock {\em J. Comput. Syst. Sci.}, 64(4):750--767, 2002.

\bibitem{Ambainis2}
Andris Ambainis.
\newblock Polynomial degree vs. quantum query complexity.
\newblock {\em J. Comput. Syst. Sci.}, 72(2):220--238, 2006.

\bibitem{BabBea}
L\'aszl\'o Babai and Robert Beals.
\newblock A polynomial-time theory of black-box groups i.
\newblock In {\em Groups St. Andrews 1997 in Bath, I}, volume 260, pages
  30--64. Cambridge Univ. Press, 1999.

\bibitem{BabSzem}
L\'aszl\'o Babai and Endre Szemer\'edi.
\newblock On the complexity of matrix group problems i.
\newblock In {\em FOCS}, pages 229--240. IEEE Computer Society, 1984.

\bibitem{BeaBab}
Robert Beals and L\'aszl\'o Babai.
\newblock Las vegas algorithms for matrix groups.
\newblock In {\em FOCS}, pages 427--436. IEEE Computer Society, 1993.

\bibitem{BBC}
Robert Beals, Harry Buhrman, Richard Cleve, Michele Mosca, and Ronald de~Wolf.
\newblock Quantum lower bounds by polynomials.
\newblock {\em J. ACM}, 48(4):778--797, 2001.

\bibitem{BBBV}
Charles~H. Bennett, Ethan Bernstein, Gilles Brassard, and Umesh Vazirani.
\newblock Strengths and weaknesses of quantum computing.
\newblock {\em SIAM J. Comput.}, 26(5):1510--1523, 1997.

\bibitem{CheMos}
Kevin K.~H. Cheung and Michele Mosca.
\newblock Decomposing finite abelian groups.
\newblock {\em Quantum Inf. Comput.}, 1(3):26--32, 2001.

\bibitem{Dam}
Ivan Damg{\aa}rd, Carmit Hazay, and Angela Zottarel.
\newblock Short paper on the generic hardness of ddh-ii.
\newblock Manuscript, \url{http://cs.au.dk/~angela/Hardness.pdf}, 2014.

\bibitem{Macaulay2}
Daniel~R. Grayson and Michael~E. Stillman.
\newblock Macaulay2, a software system for research in algebraic geometry.
\newblock Available at \url{http://www.math.uiuc.edu/Macaulay2/}.

\bibitem{Grover}
Lov~K. Grover.
\newblock Quantum mechanics helps in searching for a needle in a haystack.
\newblock {\em Physical Review Letters}, 79(2):325--328, 1997.
\newblock preliminary version in STOC 1996.

\bibitem{HLS}
Peter H\o{}yer, Troy Lee, and Robert Spalek.
\newblock Negative weights make adversaries stronger.
\newblock In David~S. Johnson and Uriel Feige, editors, {\em STOC}, pages
  526--535. ACM, 2007.

\bibitem{JouNgu}
Antoine Joux and Kim Nguyen.
\newblock Separating decision diffie-hellman from computational diffie-hellman
  in cryptographic groups.
\newblock {\em J. Cryptol.}, 16(4):239--247, 2003.

\bibitem{KanSer}
William~M. Kantor and \'Akos Seress.
\newblock {\em Black box classical groups,}, volume 208 of {\em Memoirs of the
  AMS, Vol. 208}.
\newblock American Mathematical Society, 2001.

\bibitem{LMR}
Troy Lee, Rajat Mittal, Ben~W. Reichardt, Robert Spalek, and Mario Szegedy.
\newblock Quantum query complexity of state conversion.
\newblock In Rafail Ostrovsky, editor, {\em FOCS}, pages 344--353. IEEE
  Computer Society, 2011.

\bibitem{MaurerWo99}
Ueli~M. Maurer and Stefan Wolf.
\newblock The relationship between breaking the diffie-hellman protocol and
  computing discrete logarithms.
\newblock {\em SIAM J. Comput.}, 28(5):1689--1721, 1999.

\bibitem{Mos}
Michele Mosca.
\newblock {\em Computations for Algebras and Group Representations}.
\newblock PhD thesis, University of Oxford, 1999.

\bibitem{Pohlig}
Stephen~C. Pohlig and Martin~E. Hellman.
\newblock An improved algorithm for computing logarithms over gf(p) and its
  cryptographic significance (corresp.).
\newblock {\em IEEE Trans. Inf. Theory}, 24(1):106--110, 1978.

\bibitem{Schwartz}
Jacob~T. Schwartz.
\newblock Probabilistic algorithms for verification of polynomial identities
  (invited).
\newblock In Edward~W. Ng, editor, {\em EUROSAM}, volume~72 of {\em Lecture
  Notes in Computer Science}, pages 200--215. Springer, 1979.

\bibitem{Shanks}
Daniel Shanks.
\newblock Five number-theoretic algorithms.
\newblock In {\em Proceedings of the second Manitoba conference on numerical
  mathematics}, pages 51--70, 1972.

\bibitem{Shor}
Peter~W. Shor.
\newblock Polynomial-time algorithms for prime factorization and discrete
  logarithms on a quantum computer.
\newblock {\em SIAM J. Comput.}, 26(5):1484--1509, 1997.

\bibitem{Shoup}
Victor Shoup.
\newblock Lower bounds for discrete logarithms and related problems.
\newblock In Walter Fumy, editor, {\em EUROCRYPT}, volume 1233 of {\em Lecture
  Notes in Computer Science}, pages 256--266. Springer, 1997.

\bibitem{Wat}
John Watrous.
\newblock Quantum algorithms for solvable groups.
\newblock In Jeffrey~Scott Vitter, Paul~G. Spirakis, and Mihalis Yannakakis,
  editors, {\em STOC}, pages 60--67. ACM, 2001.

\bibitem{Zha2}
Shengyu Zhang.
\newblock On the power of ambainis's lower bounds.
\newblock In Josep D\'\i{}az, Juhani Karhum\"aki, Arto Lepist\"o, and Donald
  Sannella, editors, {\em ICALP}, volume 3142 of {\em Lecture Notes in Computer
  Science}, pages 1238--1250. Springer, 2004.

\bibitem{Zippel}
Richard Zippel.
\newblock Probabilistic algorithms for sparse polynomials.
\newblock In Edward~W. Ng, editor, {\em EUROSAM}, volume~72 of {\em Lecture
  Notes in Computer Science}, pages 216--226. Springer, 1979.

\end{thebibliography}

\end{document}